\documentclass[aps,pra,twocolumn,superscriptaddress,floatfix,footinbib,showpacs,longbibliography]{revtex4-2}
\usepackage{amsmath}

\usepackage[utf8]{inputenc}  
\usepackage[T1]{fontenc}     %Output what you want e.g., é, ł, a, ü
\usepackage[table]{xcolor}
\usepackage[british]{babel}  %Do hyphenation according to british english
\usepackage[sc,osf]{mathpazo}%\linespread{1.05}  %Palatino font
\usepackage{libertineRoman}
\usepackage[colorlinks=true, citecolor=blue, urlcolor=blue]{hyperref}  
\usepackage{graphicx} % Package to insert exteral figures
\usepackage[babel]{microtype}  %Improves text justification
\usepackage{amsmath,amssymb,amsthm,bm,amsfonts,mathrsfs,bbm} %Usefull math packages

\usepackage{xspace}  %Useful to add space in macros
\usepackage{pgf,tikz}
\usepackage[dvipsnames]{xcolor}
\usepackage{multirow}
\usepackage{array}
\usepackage{bigstrut}
\usepackage{braket}
\usepackage{color}
\usepackage{natbib}
\usepackage{multirow}
\usepackage{mathtools}
\usepackage{float}
\usepackage[caption = false]{subfig}
\usepackage{xcolor,colortbl}
\usepackage{color}
\usepackage{rotating}
\usepackage{tikz}
\usepackage{ulem}
\usetikzlibrary{quantikz}
\usepackage{makecell}
\usepackage{cleveref}
\usepackage{pifont}

\newcommand{\be}{\begin{equation}}
	\newcommand{\ee}{\end{equation}}
\newcommand{\ba}{\begin{eqnarray}}
	\newcommand{\ea}{\end{eqnarray}}
\newcommand{\ketbra}[2]{|#1\rangle \langle #2|}

\newtheorem{theorem}{Theorem}
\newtheorem{corollary}{Corollary}
\newtheorem{definition}{Definition}
\newtheorem{proposition}{Proposition}

\newtheorem{lemma}{Lemma}

\def\>{\rangle}
\def\<{\langle}

\usepackage{centernot}
\usepackage{subfig}

\begin{document}

	\title{Beyond local state antidistinguishability :\\ Local state antimarking and nonlocality without entanglement}
	
	\author{Biswadeep Chatterjee}  
	\email{chatterjeebiswadeep74@gmail.com}
	\affiliation{S. N. Bose National Centre for Basic Sciences, Block JD, Sector III, Salt Lake, Kolkata 700 106, India}

	\author{Tathagata Gupta}
	\email{tathagatagupta@gmail.com}
	\affiliation{Department of Physics, Indian Institute of Technology Madras, Chennai 600036, India}
	
	\author{Pratik Ghosal}
	\email{ghoshal.pratik00@gmail.com}
	\affiliation{Harish-Chandra Research Institute, Chhatnag Road, Jhunsi, Allahabad - 211019, India}
	\affiliation{Homi Bhabha National Institute, Training School Complex, Anushakti Nagar, Mumbai 400 094, India}
	
	\author{Samrat Sen}
	\email{samrat9sen5@gmail.com}
	\affiliation{Scuola Normale Superiore, Piazza dei Cavalieri 7, 56126 Pisa, Italy}

	\begin{abstract}

		A set of quantum states is said to be antidistinguishable if, upon being given a randomly chosen state from the set, it is possible to exclude with certainty a state that was definitely not prepared. We study quantum nonlocality within the framework of local state exclusion and prove that any set of multipartite pure states containing a state that is orthogonal to all the others is locally antidistinguishable. Next, we introduce the task of \textit{local state antimarking}, where, instead of a single state, a non-repetitive sequence of states from a known set of multipartite states is randomly selected and distributed to spatially separated parties, whose objective is to identify, using LOCC, at least one sequence that was not supplied. We present an ensemble of product states that is not globally antidistinguishable; yet, by choosing states from it without replacement, one can construct sequences that can be perfectly excluded globally but, importantly, not locally. This reveals a hitherto unexplored form of `nonlocality without entanglement'. Finally, for a given ensemble, we compare its local antidistinguishability and antimarkability with its conclusive local distinguishability and markability. We demonstrate that no strict hierarchy exists between these paradigms: there exist product-state ensembles that permit one task while strictly forbidding the other, and vice versa.
	\end{abstract}

	\maketitle	
	\textit{\textbf{Introduction}:- }Traditionally, quantum nonlocality, in the form of Bell nonlocality \cite{Bell1964,BellAPS,Mermin,BellReview}, refers to input-output correlations between local measurement outcomes of spatially separated multipartite systems that defy any local-causal explanation. Entanglement is a prerequisite for exhibiting this phenomenon. However, the discovery of "quantum nonlocality without entanglement" shifted this paradigm by demonstrating that even product states can exhibit a different form of nonlocal behaviour \cite{PeresWooters,BennettNLWE}. In particular, it was shown that a multipartite quantum state chosen randomly from a known ensemble of orthogonal product states and shared among spatially separated parties may not be perfectly distinguishable via local operations and classical communication (LOCC)\footnote{It is noteworthy that although LOCC provides the natural framework for studying information-theoretic protocols under operations reflecting fundamental laboratory constraints, it is notoriously difficult to characterize mathematically \cite{Chitambar14}. This difficulty motivated the introduction of several outer approximations to LOCC, such as separable operations (SEP), a top-down approach, to bypass its mathematical complexity. In fact, the discovery of `quantum nonlocality without entanglement' proved that SEP is a strict superset of LOCC\cite{BennettNLWE}.}, thereby demonstrating a performance gap between global and local measurements.
	%In the local state discrimination (LSD) framework, where a referee randomly chooses a multipartite state and parties are restricted to retrieving encoded information via local operations and classical communication (LOCC), a set of bipartite product states demonstrates non-classical behavior, yielding a performance gap between global and local measurements. 
	Unlike Bell nonlocality, this nonlocal phenomenon is informational: while classical information can be encoded locally into these ensembles, decoding it may strictly require global measurements. The resulting gap between global and local limits has spurred significant research \cite{MassarPop, Bennett99(1), Walgate2002, DiVincenzo03, Ghosh2004, Rinaldis, Fan, Nathanson2005, Watrous, Niset06, Hayashi2006, Croke, duanCLSD, Feng, Cohen, Calsamiglia10, Gao, Bandyopadhyay2011, Bandyopadhyay, Yu, Bandyopadhyay2013, Childs2013, Cosentino2013, Fei, Tian, Chitambar14, Yang2015, Zhang2015, Xu2016, Zhang2016, Zhang2016Yongjun, Xu2017, Wang2017, Wang2017arxiv, Zhang2017Luo, Zhang2017Tan, Zhang, Zheng, Halder18, Demianowicz18, Halder19, Halder19(1), Agrawal19, Rout19, Bhattacharya20, Banik20, Rout20, LSM, Zuo2021, Zhang2021, Tathagata1, Pratik, Subhendu, Zhu2022, Subhendu2, Bandyopadhyay2024, Zhang2024, tathagata2025, Cao2025, russo2026}, owing to its foundational importance in quantum mechanics and its practical utility in cryptographic protocols \cite{Gagliardoni2021Quantum,Bennett1984Quantum,Markham2008Graph,Matthews2009Distinguishability}.
	
	A variant of the local state discrimination (LSD) task known as conclusive local state discrimination (CLSD) \cite{CheflesGlobal,Chefles}, gained importance because its primary objective is to guarantee correctness in the extraction of locally encoded information by allowing for inconclusive results. Thus, it proved highly valuable in scenarios involving locally indistinguishable states. However, because CLSD still permits inconclusive outcomes in certain rounds, the demand for reliable information extraction in \textit{every} experimental trial has recently shifted focus towards a variant of quantum state exclusion \cite{bandyopadhyay2014conclusive}, wherein the objective is to identify,  with unit probability, a state  that was definitively not prepared.  An ensemble that permits such perfect exclusion is said to be antidistinguishable \cite{Heinosaari_2018, Johnston2025}. When a set of multipartite states is antidistinguishable via LOCC, it is said to be locally antidistinguishable \cite{manna2026nonlocalityentanglementexclusionquantum}. 
	
	The task of quantum state exclusion has emerged as an important tool across diverse settings. It provides a natural framework for investigating the onticity of quantum states~\cite{Pusey2012}, exploring contextuality~\cite{srikumar}, and obtaining an operational interpretation of the Choi rank of quantum channels~\cite{stratton2024operational}. More recently, an intriguing connection between antidistinguishability and fully nonlocal correlations has been established~\cite{renner2026}. From an operational perspective, state exclusion also plays a central role in facilitating various exclusion-based tasks~\cite{uolaskr23,skr20,Uola20,skr22}, even
	demonstrating infinite quantum advantage in communication tasks~\cite{PhysRevLett.115.030504}--even in the simplest possible scenarios~\cite{PhysRevA.109.032627,bae2025randomexclusioncodesquantum,Patra2024classicalanalogueof}--as well as in establishing the communication capabilities of noisy entanglement-breaking quantum channels over classical channels~\cite{PhysRevLett.134.080803}.
	
	Motivated by this wide-ranging significance, in this work, we investigate nonlocality within the exclusion paradigm. 
	%\textcolor{blue}{The exclusion  framework  has proven valuable across diverse operational settings—from characterizing the Choi rank of a channel~\cite{stratton2024operational}, and  exploring contextuality~\cite{srikumar},  to facilitating various exclusion-based tasks~\cite{uolaskr23,skr20,Uola20,skr22}. A link between fully nonlocal correlations and antidistinguishability has also been obtained \cite{renner2026}.}..\textcolor{red}{..communication task....}. 
	We introduce a generalized variant of the local state exclusion task wherein instead of a single state,  a referee  prepares a non-repetitive sequence of the states, whose elements are chosen from a known ensemble, and distributes its subsystems to spatially separated parties. The parties must then use LOCC to identify a sequence that was definitively not supplied. If the parties were instead required to locally identify the exact non-repetitive sequence, the task is referred to as local state marking \cite{LSM}; accordingly, we name  the corresponding exclusionary variant local state antimarking (LSAM).

	\color{black}
	
	Clearly, local antimarking is inherently easier than the local state exclusion task: if an ensemble is locally antidistinguishable, any sequence containing those states can be excluded simply by excluding one of its constituents, making the ensemble locally antimarkable. Since local state discrimination enables local antidistinguishability we, therefore, have the following chain of implications:
	\begin{equation*}
		\text{LSD} \implies \text{LSAD} \implies \text{LSAM}.    
	\end{equation*}
	The central question, however, is whether the converse holds:
	\begin{equation*}
		\neg\text{LSAD} \overset{?}{\implies}  \neg\text{LSAM}.
	\end{equation*}  Here negation means an inability to perform the task which demonstrates a form of nonlocality in that setting. If the above implication holds, then the two tasks become equivalent: any set of states admitting  LSAD will necessarily admit LSAM and vice versa. If the implication, however, does not hold then the two tasks are different and we get a hierarchy of nonlocality: a locally  antidistinguishable   ensemble of states does not rule out the possibility of antimarking them, but ensembles which are nonlocal under LSAM will automatically rule out both the possibility of LSAD as well as LSD. Such ensembles will, thus, exhibit a stronger form of nonlocality than those previously reported in the literature.
	
	%\snote{Clearly, the local antimarking task is inherently less restrictive than conventional local state exclusion task; if an ensemble is locally antidistinguishable, any sequence containing those states can be excluded simply by excluding one of its constituents, \textcolor{blue}{and this ensemble is thus locally antimarkable}. The central question, however, is whether the converse holds. If it does not, then \textcolor{blue}{the two tasks are not equivalent, and we thus get a hierarchy of nonlocality; nonlocality in the antimarking task implies nonlocality in the antidistinguishability task, however nonlocality in the antidistinguishability task need not imply nonlocality in the antimarking task. Thus, ensembles displaying nonlocality under LSAM will exhibit a stronger form of nonlocality than those previously reported in the literature.}}
	
	LSAD  has recently  been studied in Ref. \cite{manna2026nonlocalityentanglementexclusionquantum}, which demonstrated that there exist certain sets of non-orthogonal product states that resist local antidistinguishability. Here, we prove a more fundamental result: for any ensemble of multipartite pure states (whether product or entangled), if there exists at least one state mutually orthogonal to all the others, the ensemble is always locally antidistinguishable. This rules out the possibility of observing nonlocality in both LSAD and LSAM for any set of mutually orthogonal multipartite pure states. Consequently, this result is as general as the seminal result of Walgate \textit{et al.} \cite{walgatehardy} for the local distinguishability framework, which established the perfect local distinguishability of any two orthogonal multipartite pure states. 
	
	Furthermore, we demonstrate that the product states introduced in Ref. \cite{manna2026nonlocalityentanglementexclusionquantum} are locally antimarkable: while they exhibit nonlocality under LSAD, an LOCC protocol can successfully exclude a permutation of the states in which the total system was not prepared. We prove the existence of ensembles that are not locally antimarkable, thereby manifesting a stronger nonlocality than mere local antidistinguishability. More significantly, we identify a peculiar `unlocking' phenomenon in certain product-state ensembles. Specifically, these states lack global antidistinguishability, rendering the question of LSAD inapplicable. However, by considering non-repetitive permutations of these states, it becomes possible to exclude a candidate ordering via global measurements, while no LOCC protocol can achieve the same. This gap confirms that these ensembles are globally, but not locally, antimarkable. They thus exhibit a form of nonlocality that is `stronger' than any previous task requiring reliable information extraction in \textit{every} experimental round, all the while remaining entirely within the product-state regime.
	
	Crucially, we also show that LSAM can differentiate between certain ensembles that appear identical under other nonlocality frameworks. For instance, a pair of state sets uniformly classified as nonlocal under conclusive local discrimination may be revealed as fundamentally distinct—one local and one nonlocal—when viewed through the lens of local antimarking. This ability to break such degeneracies provides a powerful prescription for comparing the relative nonlocality of product states, offering a characterization far more refined than any single paradigm could provide.\\

	\textit{\textbf{Antidistinguishability}}:- The task of quantum state exclusion seeks to identify which state from a known ensemble was not prepared \cite{bandyopadhyay2014conclusive}. A set of quantum states $\mathcal{S} = \{\rho_1, \rho_2, \dots, \rho_k\}$ is defined as antidistinguishable \cite{Heinosaari_2018} if there exists a Positive Operator-Valued Measure (POVM) $\mathcal{M} = \{\Pi_j\}_{j=1}^k$ satisfying two primary conditions:
	\begin{enumerate}
		
		\item \label{c1} Perfect Exclusion :
		
		$\operatorname{Tr}(\rho_j \Pi_j) = 0$ for all $j \in \{1,\dots,k\}$. This ensures that if the measurement yields outcome $j$, the observer knows with absolute certainty that the system was not prepared in state $\rho_j$.
		
		\item \label{c2} Outcome Relevance : 
		
		$\sum_{i=1}^k \operatorname{Tr}(\rho_i \Pi_j) > 0$ for all $j \in \{1,\dots,k\}$. This condition stipulates that every measurement outcome must occur with a non-zero probability for at least one state in the ensemble.
	\end{enumerate}
	While an outcome may exclude multiple states simultaneously, the requirement of both conditions  \ref{c1} and \ref{c2} ensures that there is at least one specific outcome dedicated to the exclusion of each candidate state in the set. Thus, this formulation is often referred to as strong antidistinguishability. In contrast, a more permissive definition is weak antidistinguishability \cite{stratton2024operational}, which requires only the satisfaction of the exclusion condition \ref{c1} without requiring that all outcomes be relevant \ref{c2}. Under this weaker criterion, a measurement might only be capable of excluding a proper subset of the states in $\mathcal{S}$, leaving other states in the ensemble impossible to rule out in any experimental run.  Throughout this work, by antidistinguishability, we will consider only strong antidistinguishability. Furthermore, this framework can be generalized to $m$-state exclusion, a variant where each measurement outcome must simultaneously rule out $m$ different candidate states from the ensemble \cite{stratton2024operational}.
	
	Before moving onto the main results,  we propose a sufficient condition  to ensure antidistinguishability of any set of quantum states.
	\setcounter{lemma}{0}
	\begin{lemma}\label{lemma1new}
		A set of quantum states $\mathcal{S}$ is antidistinguishable if there exist $\{\mathcal{S}_i\}_{i=1}^k$ such that $\mathcal{S} = \mathcal{S}_1\cup\mathcal{S}_2\cup\dots\cup\mathcal{S}_k$ and $\mathcal{S}_i$'s are antidistinguishable  $\forall\ i$. $\{\mathcal{S}_i\}$'s are not necessarily disjoint.
	\end{lemma}
	\begin{proof}
		Let $S_1, S_2,\dots S_k$ be antidistinguishable by the measurements $M_1, M_2,\dots M_k$ respectively, where $M_i=\{E_j^i\}_{j=1}^{|\mathcal{S}_i|}$. Now, defining the measurement as $M=\{\frac{1}{k}E_j^i\}_{i,j=1}^{k,|\mathcal{S}_i|}$, 
		completes the proof.
		% makes the set $\mathcal{S}$ antidistinguishable, thereby    
	\end{proof}

	%in Appendix \hyperref[apA]{A} of the Supplementary Material, which we shall use throughout the paper.\\ 

	\textit{\textbf{Local Antidistinguishability}}:-
	So far, no restrictions on the POVMs have been be imposed.  However, when we turn to distant laboratory settings, by default not all POVMs are allowed. The parties are restricted to perform only LOCC. Hence, the task of perfect quantum state exclusion in this setting will  be termed as local state exclusion. Let us understand this scenario better with an example. 
	
	A paradigmatic case for this study is the set of four maximally entangled Bell states in $\mathbb{C}^2 \otimes \mathbb{C}^2$, defined as:
	\begin{equation*}
		\mathcal{S}^{\text{Bell}} \equiv 
		\left\{
		\begin{aligned}
			\ket{\mathcal{B}^1} &:= \ket{\Phi^+}_{AB} := (\ket{00} + \ket{11})/\sqrt{2} \\
			\ket{\mathcal{B}^2} &:= \ket{\Phi^-}_{AB} :=(\ket{00} - \ket{11})/\sqrt{2} \\
			\ket{\mathcal{B}^3} &:= \ket{\Psi^+}_{AB} :=(\ket{01} + \ket{10})/\sqrt{2} \\
			\ket{\mathcal{B}^4} &:= \ket{\Psi^-}_{AB} :=(\ket{01} - \ket{10})/\sqrt{2}
		\end{aligned}
		\right\}.
	\end{equation*}
	While it is well-known that this ensemble cannot be perfectly or even conclusively distinguishable via LOCC, we demonstrate that it is perfectly antidistinguishable under the same local constraints.

	\begin{proposition}\label{prop:Bell_LAD}The set of four Bell states $\mathcal{S}^{\text{Bell}}$ is locally antidistinguishable.\end{proposition}
	\begin{proof}Consider a protocol where Alice and Bob each perform a local projective measurement in the computational basis $\mathcal{M} = \{|0\rangle\langle 0|, |1\rangle\langle 1|\}$. By communicating their respective outcomes, $(a, b) \in \{0, 1\}^2$, the parties can conclusively rule out specific candidate states. For instance, if both parties obtain the outcome $0$, then they know that the global state was projected onto $|00\rangle$. Since $|00\rangle$ has zero overlap with the subspace spanned by $\{|\Psi^+\rangle, |\Psi^-\rangle\}$, these states are perfectly excluded  in that run. By assigning a unique state to each of the four possible joint outcomes as detailed in Table~\ref{Table1}, the parties satisfy the criteria for strong antidistinguishability: every outcome conclusively rules out at least one state, and every state in the ensemble is excludable by at least one outcome.
	\end{proof}
	\begin{table}[b]
		\centering
		\begin{tabular}{||c|c|c||} 
			\hline
			Alice's outcome & Bob's outcome & Excluded state \\
			\hline
			$0$ & $0$ & $|\Psi^+\rangle$ \\ \hline
			$0$ & $1$ & $|\Phi^+\rangle$ \\ \hline
			$1$ & $0$ & $|\Phi^-\rangle$ \\ \hline
			$1$ & $1$ & $|\Psi^-\rangle$ \\
			\hline
		\end{tabular}
		\caption{Local antidistinguishability protocol for the set of  four Bell states $\mathcal{S}^{\text{Bell}}$. For each joint outcome $(a,b)$ obtained via local computational basis measurements, the corresponding state in the third column is conclusively excluded.}
		\label{Table1} % <--- Label must be INSIDE the table environment and AFTER the caption
	\end{table} 
	This highlights the operational leniency inherent in the paradigm of state exclusion. While the four Bell states are fundamentally nonlocal within the contexts of  perfect as well as conclusive local state discrimination \cite{CLSM}, they prove to be `perfectly local' in the setting of local state exclusion.   This transition from nonlocal to local behaviour is not unique to the ensemble of four Bell states. As we show in the following, mutual orthogonality—which ensures global distinguishability, and hence antidistinguishability—is also a sufficient condition for local antidistinguishability (See Appendix \hyperref[apB]{A} of the supplemental material for the demonstration of local antidistinguishability of the set of mutually orthogonal product states introduced by Bennett \textit{et al.}—a cornerstone of `nonlocality without entanglement').
	This connection emerges as a direct by-product of the seminal result on local state discrimination of a set of mutually orthogonal multipartite pure states, by Walgate \textit{et al.} \cite{walgatehardy}. For simplicity, we discuss the bipartite scenario for context :

	\begin{lemma}\label{lemma1}\cite{walgatehardy}
		Any two orthogonal bipartite states $|\psi\rangle_{AB}$ and $|\phi\rangle_{AB}$, shared between two distant parties Alice and Bob, can always be expressed as :
		\begin{subequations}
			\begin{align}
				|\psi\rangle_{AB} &= \sum_{i=1}^{l} |i\rangle_{A}|\eta_i\rangle_B, \\
				|\phi\rangle_{AB} &= \sum_{i=1}^{l} |i\rangle_{A}|\eta_i^\perp\rangle_B,
			\end{align}
		\end{subequations}
		where $\{|i\rangle_{A'}\}_{i=1}^l$ forms an orthogonal basis set for Alice, the vectors $\{|\eta_i\rangle_B\}$ are unnormalized and for each $i$, $|\eta_i^\perp\rangle_B$ is orthogonal to  $|\eta_i\rangle_B$. This ensures that any two  mutually orthogonal bipartite pure states are locally distinguishable. 
	\end{lemma}

	This basis enables a direct protocol for local state discrimination  of any two bipartite orthogonal states. Alice measures her subsystem in the orthogonal basis $\{|i\rangle_{A}\}$ and communicates the result, $i$, to Bob. Consequently, Bob is left with the task of locally distinguishing the states $|\eta_i\rangle_B$ and $|\eta_i^\perp\rangle_B$. Since these conditional states are guaranteed to be orthogonal, they are perfectly distinguishable. This measurement, supported by  classical communication, allows the parties to perfectly identify the initial global state

	\begin{theorem}\label{thm1}
		If for a given set of pure multipartite quantum states $\mathcal{S}=\{\ket{\psi_i}\}_{i=1}^N $ there exist a state $\ket{\psi_{j_0}}\in \mathcal{S}$ such that $\braket{\psi_i|\psi_{j_0}}=0\ \forall\ i\ne j_0$ then $\mathcal{S}$ is locally antidistinguishable.
	\end{theorem}

	\begin{proof}
		It follows from Lemma~\ref{lemma1} that given any state from a set of $N$ orthogonal states,  it is always possible to exclude at least one state via LOCC. This is because, they can run a protocol of discriminating between two arbitrary states of the ensemble, say, $\ket{\psi}$ and $\ket{\phi}$. If their outcome is $\ket{\psi}$, they can conclude that their shared state is definitely not $\ket{\phi}$, and vice versa. Note that, they can choose the pair of states from which they will perform the exclusion arbitrarily. We adopt this idea to construct local exclusion protocols for any set of orthogonal pure states in the strong sense.  
		
		Without loss of generality we take $j_0=N$ i.e. $\braket{\psi_i|\psi_{N}}=0\ \forall\ i\ne N$ therefore there exists a local protocol $\mathcal{P}_i$ to eliminate a state between $\{\ket{\psi_{i}}, \ket{\psi_{N}}\}\ \forall\ i\ne N$. Using the idea of Lemma $3$ (in the supplemental) utilizing shared randomness to select a specific protocol $\mathcal{P}_{i_0}$, the parties can perform antidistinguishability of the set of mutually orthogonal pure states. \end{proof}
	The following corollary comes naturally.

	\begin{corollary}\label{theorem1new}
		\begin{footnote}
			{This does not contradict the findings of Ref. \cite{Halder19}, because their state-discrimination objective restricts the parties to orthogonality-preserving local measurements—a constraint that is not required in our framework.} 
		\end{footnote}   Any set of mutually orthogonal pure states are locally antidistinguishable.
	\end{corollary}

	\color{black}

	Thus the above theorem clearly shows that if one expects to get nonlocality from a set of multipartite pure state ensemble, it has to be a non-orthogonal one.

	Moving forth, we shall now introduce the task of local state antimarking (LSAM). We will prove that ensembles which are not locally antimarkable exhibit a stronger form of nonlocality than those found in the local antidistinguishability paradigm.\\
	
	\textit{\textbf{Local Antimarkability}}:-
	First, we review the definition of local state marking (LSM)—originally introduced as a generalization of local state discrimination \cite{LSM}—which serves as the conceptual foundation for LSAM.

	\begin{definition}\label{defLSM}
		(Sen \textit{et  al.} \cite{LSM}). Given \( m \) states chosen randomly without repetition from a known set of $N$ mutually orthogonal multipartite quantum states \( \mathcal{S} \equiv \left\{ \ket{\psi_j} \right\}_{j=1}^N \), the \( m \)-LSM task demands correctly answering (or marking),  each of the \( m \) states via LOCC, thus figuring out the permutation of the $m$ states in the process.
	\end{definition}
	Here, $m$ ranges from 1 to $|\mathcal{S}|$, with the special case $m = 1$ corresponding to the standard local state discrimination (LSD) task. In $m$-LSM task, the parties \textit{must} correctly determine the exact sequence of the $m$ states. Since instead of discrimination, we have shifted our objective to exclusion, more importantly antidistinguishability, a natural question then arises: if the LSM task is modified to eliminate sequences instead, what new and interesting consequences might emerge? We address this by defining the local state antimarking task.

	\begin{figure}[b!]
		\centering
		\includegraphics[width=0.28\textwidth]{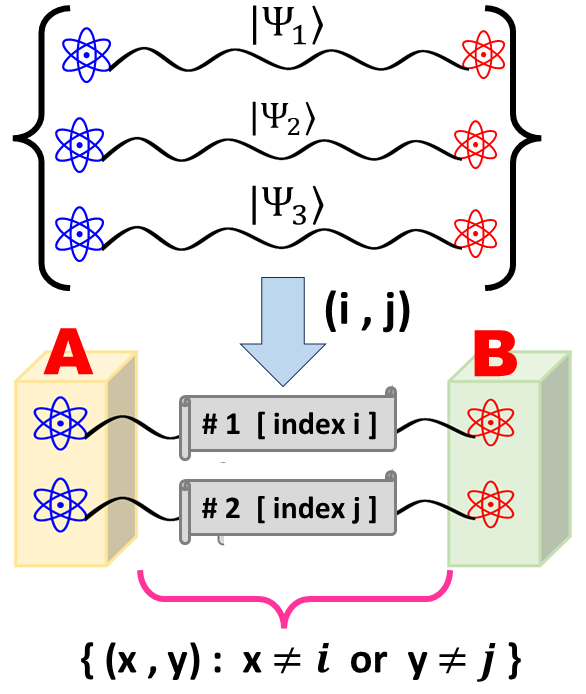}
		\caption{In the bipartite scenario, the task of $n\!=\!2$-LSAM is demonstrated for a set containing three bipartite states: $\mathcal{S} := \{\ket{\Psi_1}, \ket{\Psi_2}, \ket{\Psi_3}\}$. Two states given by index $(i,j)$, chosen randomly from the set, are distributed between spatially separated Alice and Bob without revealing their identities. The possible pairs are then : $\{\ket{\Psi_1\Psi_2}, \ket{\Psi_1\Psi_3}, \ket{\Psi_2\Psi_1}, \ket{\Psi_2\Psi_3}, \ket{\Psi_3\Psi_1}, \ket{\Psi_3\Psi_2}\}$. They must eliminate the sequence of the state shared between them via LOCC, thus answering a tuple as shown in the figure above.}
		\label{CLSMfig1}
	\end{figure}
	
	\begin{definition}
		$[(n,m)-\text{LSAM}]$ From a known ensemble of multipartite states $\mathcal{S} := \{|\psi_{i}\rangle_{A_1...A_K}\}_{i=1}^N$, a non-repetitive sequence of $n$ states $|\psi_{i_{1}}\rangle \otimes |\psi_{i_{2}}\rangle \otimes \dots \otimes |\psi_{i_{n}}\rangle$ is drawn randomly and the respective subsystems are distributed among the spatially separated parties without revealing the actual preparation indices $(i_1, i_2, \dots, i_n)$. The $(n,m)$-LSAM task requires the parties to perfectly eliminate at least $m$ sequences in every experimental run, which operationally corresponds to outputting $m$ distinct index strings $(x_1, x_2, \dots, x_n)$, each guaranteed to satisfy $(x_1, x_2, \dots, x_n) \neq (i_1, i_2, \dots, i_n)$.
	\end{definition}
	Here, just like the local state marking task, $n$ ranges from 1 to $|\mathcal{S}|$, with the special case $n = 1$ corresponding to the local$m$-state exclusion task discussed earlier.

	Intuitively, any ensemble that is locally antidistinguishable also admits local antimarking. We formalize this relationship in the following lemma.
	
	\begin{lemma}\label{lemma n to n'}
		Let $\mathcal{S}:=\{\ket{\psi_i}\}_{i=1}^N$ be an ensemble of $N$ multipartite quantum states. If $\mathcal{S}$ permits $(n,m)$-LSAM, then it also permits $(n',m')$-LSAM for any integer $n'$ such that $1 \leq n < n' \leq N$, with$$m' = m \frac{(N-n)!}{(N-n')!}$$
	\end{lemma}

	\begin{proof}
		% If $\mathcal{S}:=\{\ket{\psi_i}\}_{i=1}^N$ is $(n,m)$-LSAM, this implies there exists a LOCC protocol $\mathcal{P}_{n,m}$ to eliminate $m$ states from the set $\mathcal{S}^{[n]}:=\{\ket{\psi_{i_1}}\ket{\psi_{i_2}}\dots\ket{\psi_{i_n}}\}_{i_1,i_2,\dotsi_n=1}^N$. Now, if $n'$ states are now distributed among the parties,  all the parties apply $\mathcal{P}_{n,m}$ for the first $n$ states of the set $\mathcal{S}^{[n']}:=\{\ket{\psi_{i_1}}\ket{\psi_{i_2}}\dots\ket{\psi_{i_{n'}}}\}_{i_1,i_2,\dotsi_{n'}=1}^N$. Hence, they can eliminate $\frac{(N-n)!}{(N-n')!}m$ states in each run. Since for a given state $\ket{\Phi}=\ket{\psi_{j_1}}\ket{\psi_{j_2}}\dots\ket{\psi_{j_n}}\in \mathcal{S}^{[n]}$ there exist $\frac{(N-n)!}{(N-n')!}$ states in ${S}^{[n']}$ with the first $n$ states as $\ket{\Phi}$ so eliminating $1$ state from $\mathcal{S}^{[n]}$ implies eliminating $\frac{(N-n)!}{(N-n')!}$ states of $\mathcal{S}^{[n']}$, this concludes the proof.
		
		If the ensemble $\mathcal{S} := \{|\psi_i\rangle\}_{i=1}^N$ is $(n, m)$-LSAM, this implies there exists a LOCC protocol $\mathcal{P}_{n,m}$ that can definitively eliminate $m$ states from the set of valid length-$n$ sequences, denoted as $\mathcal{S}_{[n]} := \{|\psi_{i_1}\rangle |\psi_{i_2}\rangle \dots |\psi_{i_n}\rangle\}_{i_1, i_2, \dots, i_n = 1}^N$(where the indices are mutually distinct). Now, suppose a longer sequence of $n'$ states is distributed among the parties, which we denote as the set $\mathcal{S}^{[n']} := \{|\psi_{i_1}\rangle |\psi_{i_2}\rangle \dots |\psi_{i_{n'}}\rangle\}_{i_1, i_2, \dots, i_{n'} = 1}^N$. To perform the task, all the parties simply apply the existing protocol $\mathcal{P}_{n,m}$ exclusively to the first $n$ states of their distributed sequence. Consider a given eliminated state $|\Phi\rangle = |\psi_{j_1}\rangle |\psi_{j_2}\rangle \dots |\psi_{j_n}\rangle \in \mathcal{S}^{[n]}$. Because the states are selected without repetition, the number of ways to extend this specific sequence of length $n$ to a sequence of length $n'$ is determined by selecting the remaining $n'-n$ states from the $N-n$ available states left in $\mathcal{S}$. Thus, there exist exactly $\frac{(N-n)!}{(N-n')!}$ states in $\mathcal{S}^{[n']}$ that share $|\Phi\rangle$ as their first $n$ states. Consequently, eliminating 1 state from $\mathcal{S}^{[n]}$ automatically implies eliminating $\frac{(N-n)!}{(N-n')!}$ states from $\mathcal{S}^{[n']}$. Since the protocol $\mathcal{P}_{n,m}$ successfully eliminates $m$ distinct states from $\mathcal{S}^{[n]}$ in each run, the total number of states eliminated from $\mathcal{S}^{[n']}$ is $m' = m \frac{(N-n)!}{(N-n')!}$. 
		This concludes the proof.
	\end{proof}

	Lemma \ref{lemma n to n'} therefore proves that a set of states which is \textit{not} locally antimarkable is also \textit{not} locally antidistinguishable.
	However, the converse of the lemma is not true. 
	
	\begin{theorem}\label{theorem2new}
		A set of states which is not locally antidistinguishable, can be locally antimarkable.
	\end{theorem}
	\begin{proof}
		We show this by considering the following set of states: \begin{equation*}
			\mathcal{S}_{\text{PBR}} = \left\{ \ket{00}_{AB}, \ket{0+}_{AB}, \ket{+0}_{AB}, \ket{++}_{AB} \right\}
		\end{equation*}\label{ex-PBR} 
		
		and showing that   
		$\mathcal{S}_{PBR}$   does not allow $(1,1)$-LSAM. Hence, this set of states is not locally antidistinguishable. However, it allows $(2,3)$-LSAM. We leave the proof in Appendix \hyperref[apC]{B}, where we also show that the bipartite product states used in Ref. \cite{manna2026nonlocalityentanglementexclusionquantum} to demonstrate LSAD-nonlocality actually lose their nonlocal properties within the LSAM paradigm.  \end{proof}

	% With respect to extraction of information encoded into spatially separated systems, quantum nonlocality is a phenomenon, wherein there is a gap between the optimal global and local success strategies. 
	% The implications of the above result is the following: local antimarkability, as already mentioned before, is a relatively easier task than local antidistinguishability. Hence, if a set of states is locally antidistinguishable, it is also locally antimarkable.
	Therefore, the combination of  Lemma \ref{lemma n to n'}  and Theorem \ref{theorem2new} establishes that LSAM-nonlocality is a stronger form of nonlocality than LSAD-nonlocality. 
	
	The next logical step is to construct a set of product states that demonstrates this behaviour. This pursuit leads to the discovery of an unexpected unlocking effect.   
	
	Consider the following set of states:
	\begin{align}\label{superactivation}
		{S}_{U}\equiv\left\{
		\begin{aligned}
			\ket{\psi_1}=\ket{0}&\ket{0}\ ,\ \ket{\psi_2}=\ket{0}\ket{+},\\
			&\ket{\psi_3}=\ket{+}\ket{0}
		\end{aligned}
		\right\}
	\end{align}

	In Appendix \hyperref[apD]{C} we show that this ensemble of parent product states is not globally antidistinguishable, making it pointless to question about their local exclusion. Yet, a striking disparity emerges within the local antimarking task: the corresponding child states become globally antidistinguishable while locally they are not. This demonstrates that the local antimarking protocol can explicitly unlock a hidden, stronger form of `quantum nonlocality without entanglement' that lies asleep at the level of the parent states.

	In Appendices \hyperref[apE]{D} and \hyperref[apF]{E}, we rigorously analyse the operational relationships between the exclusion-based paradigms—local antidistinguishability  and local antimarking—and their conclusive-identification counterparts, conclusive local discrimination and conclusive local marking. By constructing explicit product-state ensembles, we show that local antimarking is capable of differentiating state ensembles that appear degenerate under other nonlocality metrics. For example, a pair of state sets uniformly classified as nonlocal under CLSD can be exposed as fundamentally distinct—one local, the other nonlocal—when subjected to the LSAM task. This capability to break nonlocality degeneracies provides a powerful new tool for comparing the relative nonclassicality of product states, offering a much more nuanced characterization than any single previous framework.\\
	
	\textit{\textbf{Discussions and Conclusion}}:-
	In this work, we have unified local state antidistinguishability  with local state marking to introduce a new operational task: local state antimarking. This framework is inherently lenient, rendering any ensemble of mutually orthogonal pure states completely local (Corollary  \ref{theorem1new}). We have also encountered an activation of nonlocality in product states through this task, in Theorem $3$ (see supplemental). Furthermore, we demonstrate the operational inequivalence between conclusive local state discrimination and local antidistinguishability, as well as between conclusive local marking and local marking (See Figs. $2$ and $3$ in the supplemental). This inequivalence allows us to resolve physical differences that a single framework might miss. Two ensembles may appear identical in one context—both being local or both nonlocal—yet diverge significantly in another. For instance, a pair of state sets  classified as nonlocal under conclusive discrimination might be revealed as one local and the other nonlocal under antidistinguishability. This ability to resolve such ambiguities provides a powerful prescription for comparing the nonlocality of various product states, offering a more refined characterization than any single paradigm could provide.
	
	Many interesting questions open up.  The concept of marking of states locally—distributing a non-repetitive sequence of states—has proven useful for establishing a hierarchy of nonlocality for state discrimination \cite{LSM,CLSM}, as well as, as shown here, for perfect state exclusion. Recently, marking of unitary operators was also used to demonstrate nonlocality without entanglement in unitary discrimination \cite{adil}. It would be interesting to explore further applications of this concept for generic quantum processes. Another intriguing question is whether there exists a set of mutually orthogonal quantum states that is not locally antidistinguishable; as per Corollary \ref{theorem1new}, any such states would necessarily have to be mixed. Secondly, our example of a set of states, which is globally antidistinguishable, yet displays nonlocality in the antimarking framework, is a tripartite example. It would be interesting to find a set of bipartite states which is globally antidistinguishable but does not allow local antimarking. Another intriguing direction for future research is to identify a set of product states that are not globally antidistinguishable in a single copy, thereby requiring a threshold of $p$ copies \cite{exclusion} to achieve global (but not local) antidistinguishability. Remarkably, under a $q$-LSAM protocol, where $q<p$ the resulting child states could become globally—but not locally—antimarkable. Such a scenario would reveal a profound operational advantage, demonstrating that the local antimarking framework can successfully extract the latent nonlocality of the parent set using a completely different and potentially more optimized number of copies than standard antidistinguishability protocols.

	\subsection*{Acknowledgment}
	BC acknowledges support from University Grants Commission, India (Reference no. 241620129062). TG acknowledges financial support from the ANRF National Post-Doctoral Fellowship (NPDF) under File No. PDF/2025/005147. PG acknowledges the support from the
	project entitled “Technology Vertical - Quantum Communication” under the National Quantum Mission of
	the Department of Science and Technology (DST) (Sanction Order No. DST/QTC/NQM/QComm/2024/2
	(G)). SS acknowledges financial support from the European Union (ERC StG ETQO, Grant Agreement no.\ 101165230). Views and opinions expressed are however those of the author(s) only and do not necessarily reflect those of the European Union or the European Research Council. Neither the European Union nor the granting authority can be held responsible for them.

	\setcounter{lemma}{3}
	%\begin{lemma}\label{lemma1new}
	% A set of quantum states $\mathcal{S}$ is antidistinguishable if there exist $\{\mathcal{S}_i\}_{i=1}^k$ such that $\mathcal{S} = \mathcal{S}_1\cup\mathcal{S}_2\cup\dots\cup\mathcal{S}_k$ and $\mathcal{S}_i$'s are antidistinguishable  $\forall\ i$. $\{\mathcal{S}_i\}$'s are not necessarily disjoint.
	%\end{lemma}
	% \begin{proof}
		%     Let $S_1, S_2,\dots S_k$ be antidistinguishable by the measurements $M_1, M_2,\dots M_k$ respectively, where $M_i=\{E_j^i\}_{j=1}^{|\mathcal{S}_i|}$. Now define the measurement $M=\{\frac{1}{k}E_j^i\}_{i,j=1}^{k,|\mathcal{S}_i|}$, this anti-distinguishes the set $\mathcal{S}$
		% \end{proof}
	
	\onecolumngrid
	\section*{Supplementary Material}
	
	Before moving on to the technical parts of the paper, it is instructive to consider the exclusion of three pure states, a scenario for which  necessary and sufficient conditions are known.  As established by Caves \textit{et al.} \cite{CarltonCaves} (see also \cite{barrett, Molina2019povmsareequivalent} ), an ensemble of three pure states $\mathcal{S} = \{|\psi_1\rangle, |\psi_2\rangle, |\psi_3\rangle\}$ is strongly antidistinguishable if and only if their pairwise overlaps $x_{12} = |\langle\psi_1|\psi_2\rangle|^2$, $x_{13} = |\langle\psi_1|\psi_3\rangle|^2$ and $x_{23} = |\langle\psi_2|\psi_3\rangle|^2$ satisfy the conditions
	\begin{subequations}
		\begin{align}
			x_{12} + x_{13} + x_{23} &< 1, \label{eq:caves1} \\  
			[1 - (x_{12} + x_{13} + x_{23})]^2 &\geq 4x_{12} x_{13} x_{23}. \label{eq:caves2}
		\end{align}
	\end{subequations}

	For example, consider the symmetric trine states, 
	\begin{equation}
		|T_k\rangle = \cos\left[\frac{2\pi}{3}(k-1)\right]|0\rangle + \sin\left[\frac{2\pi}{3}(k-1)\right]|1\rangle,
	\end{equation}
	for $k \in \{1,2,3\}$. This geometrically uniform ensemble $\mathcal{S}_{Tri} \equiv \{|T_1\rangle, |T_2\rangle, |T_3\rangle\}$ is strongly antidistinguishable. The optimal exclusion strategy uses the measurement $\mathcal{M}_{Tri} \equiv \left\{\frac{2}{3}|T_k^\perp\rangle\langle T_k^\perp|\right\}_{k=1}^3$, where $|T_k^\perp\rangle$ denotes the state orthogonal to $|T_k\rangle$. In this configuration, every measurement outcome occurs with non-zero probability across the ensemble, and the $k$-th outcome conclusively eliminates $|T_k\rangle$. Thus, no outcome is redundant, and every state is perfectly excludable.
	\section*{Appendix A: Bennett product states are local in LSAD}\label{apB}
	The nine orthogonal product states of Bennett \textit{et al.} \cite{BennettNLWE} are:
	
	\renewcommand{\arraystretch}{1.3}
	\begin{equation}
		\mathcal{S}^{\text{B}} \equiv \left\{
		\begin{array}{ll}
			\ket{\psi}_1 = \ket{2}_A\ket{1{-}2}_B,~\ket{\psi}_2 = \ket{1{+}2}_A\ket{0}_B,\\
			\ket{\psi}_3 = \ket{1{-}2}_A\ket{0}_B,~\ket{\psi}_4 = \ket{0}_A\ket{0{+}1}_B,\\
			\ket{\psi}_5 = \ket{0}_A\ket{0{-}1}_B,~\ket{\psi}_6 = \ket{2}_A\ket{1{+}2}_B,\\
			\ket{\psi}_7 = \ket{0{+}1}_A\ket{2}_B,~\ket{\psi}_8 = \ket{0{-}1}_A\ket{2}_B,\\
			\hspace{1.5cm}\ket{\psi}_9 = \ket{1}_A\ket{1}_B \\
		\end{array}
		\right\},
	\end{equation}
	\renewcommand{\arraystretch}{1.0}

	where we define $\ket{p \pm q} := \frac{\ket{p} \pm \ket{q}}{\sqrt{2}}$, omitting normalization factors for brevity.
	\setcounter{proposition}{1}
	\begin{proposition}\label{p2}
		The set $\mathcal{S}^{\text{B}}$ is locally antidistinguishable.
	\end{proposition} 
	\begin{proof}
		Alice performs the measurement $\mathcal{M}_A\equiv\{\ketbra{0}{0},\ketbra{1}{1},\ketbra{2}{2}\}$ and Bob performs the measurement $\mathcal{M}_B\equiv\{\ketbra{0}{0},\ketbra{1}{1},\ketbra{2}{2}\}$. By communicating their outcomes eliminate the states as given in table \ref{Table2}.
	\end{proof}

	\begin{table}[t]
		\centering
		
		\begin{tabular}{||c|c|c||}
			\hline
			Alice's outcome& Bob's outcome & Eliminated states \\
			\hline
			$\ketbra{0}{0}$ & 
			\begin{tabular}{c}
				$\ketbra{0}{0}$\\
				$\ketbra{1}{1}$\\
				$\ketbra{2}{2}$
			\end{tabular}
			& \begin{tabular}{c}
				$\ket{\psi}_1$\\
				$\ket{\psi}_2$\\
				$\ket{\psi}_3$
			\end{tabular} \\
			\hline
			
			$\ketbra{1}{1}$ & 
			\begin{tabular}{c}
				$\ketbra{0}{0}$\\
				$\ketbra{1}{1}$\\
				$\ketbra{2}{2}$
			\end{tabular}
			& \begin{tabular}{c}
				$\ket{\psi}_4$\\
				$\ket{\psi}_5$\\
				$\ket{\psi}_6$
			\end{tabular} \\
			\hline
			
			$\ketbra{2}{2}$ & 
			\begin{tabular}{c}
				$\ketbra{0}{0}$\\
				$\ketbra{1}{1}$\\
				$\ketbra{2}{2}$
			\end{tabular}
			& \begin{tabular}{c}
				$\ket{\psi}_7$\\
				$\ket{\psi}_8$\\
				$\ket{\psi}_9$
			\end{tabular} \\
			\hline
		\end{tabular}
		\caption{The detailed protocol for the local antidistinguishability of  $\mathcal{S}^{\text{B}} $ is shown. For example, if the outcome corresponding to $\ketbra{0}{0}$ clicks on Alice's side, and $\ketbra{0}{0}$ clicks on Bob's side, they eliminate the state $\ket{\psi}_5$. }\label{Table2}
	\end{table}

	Corollary $1$ establishes that any ensemble of mutually orthogonal multipartite pure states is local under the LSAD task. This clearly demonstrates that non-orthogonality is a strict prerequisite for manifesting any form of nonlocality in this context.

	\section*{Appendix B : Proof of Theorem $2$}\label{apC} 
	
	\noindent\textbf{Theorem 2}:
	\textit{   Nonlocality in LSAM represents a stronger form of quantum nonlocality than its counterpart in LSAD.}
	\begin{proof}

		Consider the following set of states: \begin{equation*}
			\mathcal{S}_{\text{PBR}} = \left\{ \ket{00}_{AB}, \ket{0+}_{AB}, \ket{+0}_{AB}, \ket{++}_{AB} \right\}
		\end{equation*}\label{ex-PBR}

		Since the local state sets of both Alice and Bob is $\{\ket{0},\ket{+}\}$ are not antidistinguishable which assures that $\mathcal{S}_{\text{PBR}}$ is not locally antidistinguishable\cite{manna2026nonlocalityentanglementexclusionquantum} i.e. not (1,1)-LSAM . 
		
		Now, let us consider the following set of states, for the $n=2$ antimarking task. 
		\begin{equation}
			\mathcal{S}^{[2]}_{\text{PBR}} = \left\{
			\begin{aligned}
				&\ket{00}\ket{0+}, &&\ket{0+}\ket{00}, \\
				&\ket{0+}\ket{0+}, &&\ket{00}\ket{+0}, \\
				&\ket{0+}\ket{+0}, &&\ket{0+}\ket{++}, \\
				&\ket{+0}\ket{00}, &&\ket{+0}\ket{0+}, \\
				&\ket{++}\ket{0+}, &&\ket{+0}\ket{+0}, \\
				&\ket{+0}\ket{++}, &&\ket{++}\ket{+0}
			\end{aligned}
			\right\}_{A_1A_2|B_1B_2}
		\end{equation}
		Alice and Bob perform the projective measurement $\mathcal{M}_{\text{PBR}} = \{ \ket{\xi_k}\bra{\xi_k} \}_{k=1}^4$ on their respective systems $A_1A_2$ and $B_1B_2$, where 
		\begin{equation}
			\begin{aligned}
				\ket{\xi_1} &= \tfrac{1}{\sqrt{2}}(\ket{01} + \ket{10}), & \ket{\xi_2} &= \tfrac{1}{\sqrt{2}}(\ket{0-} + \ket{1+}), \\
				\ket{\xi_3} &= \tfrac{1}{\sqrt{2}}(\ket{+1} + \ket{-0}), & \ket{\xi_4} &= \tfrac{1}{\sqrt{2}}(\ket{+-} + \ket{-+}).
			\end{aligned}
		\end{equation}
		It is straightforward to verify that for every joint outcome $(k, l)$, the parties can eliminate at least three states from the ensemble.  \end{proof}

	We now show that the product states used in Ref. \cite{manna2026nonlocalityentanglementexclusionquantum} to demonstrate LSAD-nonlocality actually lose their nonlocal properties within the LSAM paradigm.

	\begin{proposition}\label{ex-theta}
		The ensemble is defined as $\mathcal{S}_{\theta} = \{ \ket{\psi_1}_{AB}, \ket{\psi_2}_{AB}, \ket{\psi_3}_{AB}, \ket{\psi_4}_{AB} \}$, where:
		\begin{equation}
			\begin{aligned}
				\ket{\psi_1}_{A|B} &= \ket{\theta_+}_A\ket{\theta_+}_B, & \ket{\psi_2}_{A|B} &= \ket{\theta_+}_A\ket{\theta_-}_B, \\
				\ket{\psi_3}_{A|B} &= \ket{\theta_-}_A\ket{\theta_+}_B, & \ket{\psi_4}_{A|B} &= \ket{\theta_-}_A\ket{\theta_-}_B,
			\end{aligned}
		\end{equation}
		with $\ket{\theta_\pm} = \cos\theta\ket{0} \pm \sin\theta\ket{1}$. For $\cos2\theta \leq \sqrt{2}-1$, the set $\mathcal{S}$ is not $(1,2)$-LSAM\cite{manna2026nonlocalityentanglementexclusionquantum}, yet the ensemble $\mathcal{S}^{[2]}$ admits $(2,8)$-LSAM.
	\end{proposition}
	
	\begin{proof}
		Consider the $n=2$ antimarking scenario. The set of child states corresponds to the distinct permutations $(i,j)$ for $i \neq j$. By defining the permutation states as $\ket{\Psi_{i,j}}_{A_1A_2|B_1B_2}$ (which represents $\ket{\psi_i}_{A_1B_1} \otimes \ket{\psi_j}_{A_2B_2}$ rearranged to the local partition), the set $\mathcal{S}^{[2]}_{\theta}$ associated with the partition $A_1A_2|B_1B_2$ is defined as follows:
		\begin{equation}
			\mathcal{S}^{[2]}_{\theta} = \left\{
			\begin{aligned}
				\ket{\Psi_{1,2}} &= \ket{\theta_+\theta_+}\ket{\theta_+\theta_-}, \\
				\ket{\Psi_{2,1}} &= \ket{\theta_+\theta_+}\ket{\theta_-\theta_+}, \\
				\ket{\Psi_{1,4}} &= \ket{\theta_+\theta_-}\ket{\theta_+\theta_-}, \\
				\ket{\Psi_{1,3}} &= \ket{\theta_+\theta_-}\ket{\theta_+\theta_+}, \\
				\ket{\Psi_{2,3}} &= \ket{\theta_+\theta_-}\ket{\theta_-\theta_+}, \\
				\ket{\Psi_{2,4}} &= \ket{\theta_+\theta_-}\ket{\theta_-\theta_-}, \\
				\ket{\Psi_{3,1}} &= \ket{\theta_-\theta_+}\ket{\theta_+\theta_+}, \\
				\ket{\Psi_{3,2}} &= \ket{\theta_-\theta_+}\ket{\theta_+\theta_-}, \\
				\ket{\Psi_{4,1}} &= \ket{\theta_-\theta_+}\ket{\theta_-\theta_+}, \\
				\ket{\Psi_{4,2}} &= \ket{\theta_-\theta_+}\ket{\theta_-\theta_-}, \\
				\ket{\Psi_{3,4}} &= \ket{\theta_-\theta_-}\ket{\theta_+\theta_-}, \\
				\ket{\Psi_{4,3}} &= \ket{\theta_-\theta_-}\ket{\theta_-\theta_+}
			\end{aligned}
			\right\}_{A_1A_2|B_1B_2}
		\end{equation}
		
		Since the set $\mathcal{S}^{[2]}_{\theta}$ is globally $2$-antidistinguishable \cite{PhysRevResearch.5.023094} from the $A_1A_2|B_1B_2$ partition of $S_{PBR}^{[2]}$, we can see the local systems of both parties are $\{\ket{\theta_+}\ket{\theta_+}, \ket{\theta_+}\ket{\theta_-}, \ket{\theta_-}\ket{\theta_+}, \ket{\theta_-}\ket{\theta_-}\}$. Let Alice start the protocol: by performing this global measurement $\{\Pi_A,\Pi_B,\dots,\Pi_F\}$ as given below,
		\begin{align}
			\Pi_A &= \gamma \, \ketbra{+\bar{\theta},0}{+\bar{\theta},0}, \\
			\Pi_B &= \gamma \, \ketbra{0,+\bar{\theta}}{0,+\bar{\theta}}, \\
			\Pi_C &= \gamma \, \ketbra{0,-\bar{\theta}}{0,-\bar{\theta}}, \\
			\Pi_D &= \gamma \, \ketbra{-\bar{\theta},0}{-\bar{\theta},0}, \\
			\Pi_E &= \alpha \, \ketbra{\psi_{01}^{+}}{\psi_{01}^{+}} 
			+ \beta \, \ketbra{\psi_{cs}^{+}}{\psi_{cs}^{+}}, \\
			\Pi_F &= \alpha \, \ketbra{\psi_{01}^{-}}{\psi_{01}^{-}} 
			+ \beta \, \ketbra{\psi_{cs}^{-}}{\psi_{cs}^{-}}.
		\end{align}
		where, $\beta =\frac{1}{2\cos^4\theta}\ ,\ \gamma=\frac{1-\tan^4\theta}{4\sin^2\theta}\ ,\ \alpha = \frac{1}{2}-\gamma\cos^2\theta\ ,\ \ket{\pm\bar{\theta}}=\sin\theta\ket{0}\mp\cos\theta\ket{1},\ket{\psi^\pm_{01}}=\ket{00}\pm\ket{11}\ \text{and}\ \psi_{cs}^\pm=\sin^2\theta\ket{0}\pm\cos^2\theta\ket{1}$
		For each outcome Alice eliminates the states that belong to one of the respective pairs
		\begin{align}
			A &= \{++, +-\}, & B &= \{++, -+\}, \\
			C &= \{+-, --\}, & D &= \{-+, --\}, \\
			E &= \{+-, -+\}, & F&= \{++, --\}.
		\end{align}
		Following this she can always eliminate $6$ states from $S_{\theta}^{[2]}$. Further, Bob also performs the same global measurement on his end and can eliminate at least two more states. Thus, $S_{\theta}$ is $(2,8)$-LSAM.
		
	\end{proof}

	Much like the previous examples, the set \begin{align}\label{NL_1}
		\mathcal{S}_{\text{NL}_1} \equiv \{\ket{0}_A\ket{+}_B,\ket{+}_A\ket{0}_B, \ket{i_+}_A\ket{i_+}_B\} 
	\end{align}
	suffers the same fate: although it showed nonlocality in the LSAD paradigm, it is in fact locally antimarkable.
	
	\begin{proposition}\label{prop7}
		The set of states $\mathcal{S}_{NL_1}$  allows $(2,1)$-LSAM.
	\end{proposition} 
	\begin{proof}
		Since, the local parts $\mathcal{S}_{N_1}^{[2]}$ are:
		\begin{align}
			(S_{NL_1}^{[2]})_A=(S_{NL_1}^{[2]})_B = \left\{ 
			\begin{aligned}
				&\ket{\psi_1}=\ket{0+}, \quad
				\ket{\psi_2}=\ket{0i_+}, \\
				&\ket{\psi_3}=\ket{+0}, \quad
				\ket{\psi_4}=\ket{+i_+}, \\
				&\ket{\psi_5}=\ket{i_+\,0}, \quad
				\ket{\psi_6}=\ket{i_+\,+}
			\end{aligned}
			\right\}
		\end{align}
		Now, define the subsets:
		\begin{align}
			S_1 &= \{ \ket{\psi_1}, \ket{\psi_4}, \ket{\psi_5} \}\\
			S_2 &= \{ \ket{\psi_2}, \ket{\psi_3}, \ket{\psi_6} \} 
		\end{align}
		Both $S_1\ \&\ S_2$ are antidistinguishable since they satisfy the \cref{eq:caves1,eq:caves2} therefore by Lemma 1 % 1 
		$(S_{NL_1}^{[2]})_A\ \&\ (S_{NL_1}^{[2]})_B$ are antidistingusihable which implies just by measuring on there local parts Alice or Bob can antidistinguish $S_{NL_1}^{[2]}$ i.e. $S_{NL_1}$ is $(2,1)$-antimarkable.\end{proof}

	\section*{Appendix C: Unlocking of nonlocality in product states through LSAM}\label{apD}
	Consider the following set of states:
	\begin{align}\label{superactivation}
		{S}_{U}\equiv\left\{
		\begin{aligned}
			\ket{\psi_1}=\ket{0}&\ket{0}\ ,\ \ket{\psi_2}=\ket{0}\ket{+},\\
			&\ket{\psi_3}=\ket{+}\ket{0}
		\end{aligned}
		\right\}
	\end{align}
	\setcounter{theorem}{2}
	\begin{theorem}\label{P_SuperActivation}
		The set $\mathcal{S}_{U}$ is not globally antidistinguishable. However, it shows nonlocality in the $(2,1)$-LSAM task.
	\end{theorem}
	\begin{proof}
		Using \cref{eq:caves1,eq:caves2}, it is straightforward to see that $\mathcal{S}_{U}$ is not globally antidistinguishable. To prove that $\mathcal{S}_{U}$ shows nonlocality in the  $(2,1)$-LSAM task, we consider
		\begin{align}
			\mathcal{S}^{[2]}_{U} =
			\left\{
			\begin{aligned} 
				\ket{\psi_{12}} &= \ket{00}\ket{0+},\ \ket{\psi_{13}} = \ket{0+}\ket{00},\\ 
				\ket{\psi_{21}} &= \ket{00}\ket{+0},\ \ket{\psi_{23}} = \ket{0+}\ket{+0},\\ 
				\ket{\psi_{31}} &= \ket{+0}\ket{00},\ \ket{\psi_{32}} = \ket{+0}\ket{0+}
			\end{aligned}
			\right\}_{A_1A_2|B_1B_2}
		\end{align}
		
		We define the subsets $\mathcal{S}_1 = \{\ket{\psi_{12}}, \ket{\psi_{13}}, \ket{\psi_{21}}\},\quad
		\mathcal{S}_2 = \{\ket{\psi_{12}}, \ket{\psi_{23}}, \ket{\psi_{31}}\}\ \&\ \mathcal{S}_3 = \{\ket{\psi_{13}}, \ket{\psi_{21}}, \ket{\psi_{32}}\}$. Each of these subsets is antidistinguishable, as they satisfy conditions \cref{eq:caves1,eq:caves2}. Therefore, by Lemma 1, it follows that $\mathcal{S}_{U}^{[2]}$ is antidistinguishable via global measurements.
		
		% , i.e. $\mathcal{S}_{SA}$ can be globally  $(2,1)$-antimarked.
		
		Finally, observe that the local sets associated with $\mathcal{S}_{U}^{[2]}$ are $\{\ket{00}, \ket{0+}, \ket{+0}\}$ for both Alice and Bob. Since none of these local sets are antidistinguishable via  Eqns. \cref{eq:caves1,eq:caves2} and therefore  it follows that $\mathcal{S}_{U}$ is not $(2,1)$-LSAM.
	\end{proof}

	\section*{Appendix D: CLSD vs LSAD}\label{apE}
	
	Let us briefly discuss the task of conclusive state discrimination (CSD), which will help us understand its local variant CLSD.
	
	Consider a quantum state $\ket{\psi_k}$ drawn randomly from a known ensemble $\mathcal{S}\equiv\left\{ \ket{\psi_1}, \ket{\psi_2}, \dots, \ket{\psi_N} \right\}$. The objective of CSD is the unambiguous identification of the prepared state with zero probability of error, achieved at the expense of allowing a non-zero probability of obtaining an inconclusive outcome. This protocol thus  requires a Positive Operator-Valued Measure (POVM) characterized by $N+1$ outcomes, represented as $\mathcal{M}\equiv\{P_?, P_k \mid P_? + \sum_{k=1}^N P_k = \mathbb{1}_\mathcal{H}\}$. Here, the element $P_?$ represents the inconclusive result, while the remaining $N$ elements correspond to the definitive identification of each state in the ensemble. To enforce the zero-error constraint, the POVM elements must strictly satisfy the orthogonality condition:$$\operatorname{Tr}(P_{k} \ket{\psi_l}\bra{\psi_l}) = p_k \delta_{l,k}$$where $p_k = \operatorname{Tr}(P_k \ket{\psi_k}\bra{\psi_k}) > 0$ denotes the success probability of correctly detecting the state $\ket{\psi_k}$.
	Subsequently, a necessary and sufficient condition for conclusive distinguishability is obtained. \begin{lemma}\label{lemmaCheflesGlobal}
		[Chefles \cite{CheflesGlobal}] A set of 
		$N$ quantum states can be conclusively distinguished if and only if the states are linearly independent. 
	\end{lemma}
	For instance, the set \( \{\ket{0}, \ket{+}, \ket{1}\} \subset \mathbb{C}^2 \) cannot be conclusively distinguished. Interestingly, within the LOCC paradigm the concept of conclusive distinguishability gets modified due to constraints on the allowed operations.
	\setcounter{definition}{2}
	\begin{definition}\label{defCheflesCLSD}
		[Chefles \cite{Chefles}] A multipartite state \( |\psi\rangle \in S \) is conclusively locally identifiable if and only if there is an 
		LOCC protocol whereby with some nonzero probability \( p > 0 \) it can be determined that the multipartite quantum system \( Q \) 
		was certainly prepared in state \( |\psi\rangle \). The set \( S \) is conclusively locally distinguishable if and only if every state  in \( S \) is conclusively locally identifiable.
	\end{definition}
	
	Having established the operational framework of CLSD, we now compare it with LSAD to determine whether a strict hierarchical relationship exists between these two paradigms.
	\color{black}

	Consider the following set of states:

	\begin{align}
		\mathcal{S}_{\text{NL}}^\theta \equiv \left\{
		\begin{aligned}
			&\ket{\psi_1} := \ket{0}_A \ket{0}_B, \quad 
			\ket{\psi_2} := \ket{\theta}_A \ket{\theta}_B, \\
			&\ket{\psi_3} := \ket{+}_A \ket{+}_B,\, \ket{\psi_4} := \ket{i_+}_A \ket{i_-}_B
		\end{aligned}
		\right\}, \label{notCLSDnotLSAD}
	\end{align}
	
	where $\ket{\pm}:=1/\sqrt{2}(\ket{0} \pm\ket{1})$, $\ket{i_{\pm}}:=1/\sqrt{2}(\ket{0} \pm\iota\ket{1})$, with $\iota:=\sqrt{-1}$ and $\ket{\theta}=\cos(\theta)\ket{0}+\sin(\theta)\ket{1}$

	\begin{proposition}\label{Appendix_D_Prop_5}
		The set $\mathcal{S}_{\text{NL}}^{\frac{4\pi}{5}}$ is globally antidistinguishable and conclusively disitinguishable. However, it is neither locally antidistinguishable nor conclusively disitinguishable.
	\end{proposition}

	\begin{proof}
		To establish strong global antidistinguishability, consider the sets $\mathcal{S}_1:=\{\ket{\psi_1},\ket{\psi_3},\ket{\psi_4}\}$ and $\mathcal{S}_2:=\{\ket{\psi_1},\ket{\psi_2},\ket{\psi_3}\}$. One can trivially verify that these sets satisfy Caves' conditions (\ref{eq:caves1}) and (\ref{eq:caves2}). By invoking Lemma 1 for their union, $\mathcal{S}_1 \cup \mathcal{S}_2 = \mathcal{S}_{\text{NL}}^{\frac{4\pi}{5}}$, we complete the proof and conclude that $\mathcal{S}_{\text{NL}}^{\frac{4\pi}{5}}$ is strongly globally antidistinguishable.  However, the set of states is not locally antidistinguishable. For Alice's local system, represented by the states $\{\ket{0}, \ket{\frac{4\pi}{5}}, \ket{+}, \ket{i_+} \}$, there do not exist coefficients $\alpha_i \geq 0$ such that $\alpha_1\ketbra{0}{0}+\alpha_2\ketbra{\frac{4\pi}{5}}{\frac{4\pi}{5}}+\alpha_3\ketbra{+}{+}+\alpha_4\ketbra{i_+}{i_+}=\mathbb{I}$. Thus, this set of Alice's marginal states is not even weakly antidistinguishable \cite{Heinosaari_2018}. The same argument applies to Bob's local system. By Theorem 2 of \cite{manna2026nonlocalityentanglementexclusionquantum}, this implies that $\mathcal{S}_{\text{NL}}^{\frac{4\pi}{5}}$ is not locally weakly antidistinguishable, and is therefore not locally antidistinguishable. Since $\mathcal{S}_{\text{NL}}^{\frac{4\pi}{5}}$ is a set of linearly independent states, Lemma \ref{lemmaCheflesGlobal} dictates that it is globally conclusively distinguishable. However, because $\{\ket{0},\ket{+}\}$ spans $\mathbb{C}^2$, there does not exist any state $\ket{a}_A\ket{b}_B$ such that $\braket{ab\vert{}\psi_i}=0$ for $i=1,3,4$. Therefore, based on Definition \ref{defCheflesCLSD}, we conclude that $\mathcal{S}_{\text{NL}}^{\frac{4\pi}{5}}$ is not locally conclusively distinguishable.
	\end{proof}

	We shall also show that the CLSD task also acts as a diagnostic tool to distinguish between two sets of states that appear identical—or local—within the LSAD paradigm.

	To see this, consider the following set of states:

	\begin{align}\label{NL_1}
		\mathcal{S}_{\text{NL}_1} \equiv \{\ket{0}_A\ket{+}_B,\ket{+}_A\ket{0}_B, \ket{i_+}_A\ket{i_+}_B\} 
	\end{align}

	\begin{proposition}
		The state ensemble $\mathcal{S}_{\text{NL}_1}$ permits CLSD. Furthermore, despite being globally antidistinguishable, the states are not  antidistinguishable locally.
	\end{proposition}
	
	\begin{proof}
		The global antidistinguishability of the set $\mathcal{S}_{NL_1}$ is ensured since they satisfy the if and only if conditions \ref{eq:caves1} and \ref{eq:caves2}. However, the local system of each party is the set $\{\ket{0}, \ket{+}, \ket{+i}\}$ which violates condition \ref{eq:caves1} and thus are not antidistinguishable, the non-antidistinguishability of the local system ensures the local non-antidistinguishability of $\mathcal{S}_{NL_1}$ since they are product states \cite{manna2026nonlocalityentanglementexclusionquantum}. On the other hand , both the local parties perform the following measurement:   
		\begin{equation}
			\begin{aligned}
				\mathcal{M} \equiv \big\{ &M_1 = \tfrac{1}{3}\ketbra{1}{1}, M_2 = \tfrac{1}{3}\ketbra{-}{-}, \\
				&M_3 = \tfrac{1}{3}\ketbra{i_-}{i_-}, M_4 = I - \sum_{j=1}^3 M_j \big\}
			\end{aligned}
		\end{equation}
		
		The parties then classically communicate to each other and tally their outcomes. For example, an outcome corresponding to the projector onto $11$ conclusively distinguishes the state $\ket{i_+i_+}$. Hence, like this each and every member of  $\mathcal{S}_{NL_1}$ can be conclusively distinguished.
	\end{proof}

	Let us now consider another set of states:  the anti-parallel double-SIC ensemble \cite{CLSM}
	\begin{align}\label{anti-SIC}
		\mathcal{S}^{\uparrow\hspace{-.05cm}\downarrow}:=\left\{
		\ket{\gamma_i}:=\ket{s_i}_A\otimes\ket{s^\perp_i}_B~\text{s.t.}~\ket{s_i}\in\mathrm{SIC}
		\right\},  
	\end{align}
	where \(\ket{\psi^\perp}\) denotes the state orthogonal to \(\ket{\psi}\) and 
	SIC or the symmetric informationally complete ensemble of the qubit refers to:
	\begin{align}\label{SIC}
		\mathrm{SIC}
		\equiv \left\{\begin{aligned}
			&~~|s_1\rangle:= |0\rangle,~\text{and ~for}~j\in\{2,3,4\}\\
			&|s_j\rangle:= \frac{1}{\sqrt{3}} \left( |0\rangle + e^{2\pi\iota (j-2)/3}\sqrt{2} |1\rangle \right)
		\end{aligned} \right\}.   
	\end{align}
	
	Mirroring the behavior of $\mathcal{S}_{NL_1}$ in the antidistinguishability scenario, this set exhibits no nonlocality, as it is both globally as well as local antidistinguishable.
	
	\begin{proposition}\label{antisic}
		The set $\mathcal{S}^{\uparrow\hspace{-.05cm}\downarrow}$ is both globally and locally antidistinguishable.
		
	\end{proposition}
	\begin{proof}
		
		To see this, either Alice or Bob starts the protocol. If Alice starts it, she performs the measurement $\mathcal{M}:\{\frac{1}{2}\ket{s^\perp_i}\bra{s^\perp_i}\}_i$. If she obtains an outcome corresponding to the projector $\ket{s^\perp_{i^*}}\bra{s^\perp_{i^*}}\}$, she tells Bob her outcome and both of them declare, that definitely the state  $\ket{\gamma_{i^*}}$ was not prepared and trivially the two parties are able to locally antidistinguish the states.
	\end{proof}

	However, the set $\mathcal{S}^{\uparrow\hspace{-.05cm}\downarrow}$ is globally conclusively distinguishable but not locally \cite{CLSM} which shows that CLSD task is able to make a distinction between two seemingly local set of states. Hence, this proves that the CLSD and the LSAD paradigms are inequivalent.

	\section*{Appendix E: CLSM vs LSAM}\label{apF}
	
	Moving on, we point out that having established the inequivalence of the LSAD and CLSD frameworks through the examples $\mathcal{S}_{\text{NL}_1}$ and $\mathcal{S}^{\uparrow\hspace{-.05cm}\downarrow}$, it is now instructive to examine the relationship of LSAM with the conclusive state marking (CLSM) paradigm \cite{CLSM}, wherein states which do not allow CLSM are proven to  show a stronger form of nonlocality than CLSD.

	We first start by revisiting the definition of the CLSM task. 
	\begin{definition}\label{def3}
		[\(m\)-CLSM] Given \( m \) states chosen randomly from a known set of multipartite quantum states \( \mathcal{S} \equiv \left\{ \ket{\psi_j} \right\}_{j=1}^N \), the \( m \)-CLSM task demands marking each of the \( m \) states via LOCC, conclusively.
	\end{definition}

	We have already shown in Proposition \ref{Appendix_D_Prop_5} that the set of product states $\mathcal{S}_{\text{NL}}^{\frac{4\pi}{5}}$ is CLSD nonlocal. Now, we shall show that it is also $2$-CLSM nonlocal. However, even though this set of states showed nonlocality in the LSAD paradigm, we find that it loses its nonlocality in the LSAM paradigm. 
	
	\begin{equation}
		(\mathcal{S}_{NL}^{\frac{4\pi}{5}})^{[2]} \equiv
		\left\{
		\begin{aligned}
			\ket{\Psi_1} &= \ket{0\theta}\ket{0\theta},   & \ket{\Psi_2} &= \ket{0+}\ket{0+}, \\
			\ket{\Psi_3} &= \ket{0i_+}\ket{0i_-},      & \ket{\Psi_4} &= \ket{\theta0}\ket{\theta0}, \\
			\ket{\Psi_5} &= \ket{\theta+}\ket{\theta+},   & \ket{\Psi_6} &= \ket{\theta i_+}\ket{\theta i_-}, \\
			\ket{\Psi_7} &= \ket{+0}\ket{+0},        & \ket{\Psi_8} &= \ket{+\theta}\ket{+\theta}, \\
			\ket{\Psi_9} &= \ket{+i_+}\ket{-i_+},      & \ket{\Psi_{10}} &= \ket{i_+0}\ket{i_-0}, \\
			\ket{\Psi_{11}} &= \ket{i_+\theta}\ket{i_-\theta}, & \ket{\Psi_{12}} &= \ket{i_++}\ket{i_-+}
		\end{aligned}
		\right\}
	\end{equation}
	
	\begin{proposition}\label{Appendix_E_Prop_8}
		The set  $\mathcal{S}_{\text{NL}}^{\frac{4\pi}{5}}$  shows no nonlocality in the $(2,1)-$LSAM task. However, in the $2-$CLSM scenario, it shows nonlocality.
	\end{proposition}

	\begin{proof}
		To see this, first observe that both Alice has the following set of marginal states:
		\begin{equation}
			(\mathcal{S}_{NL}^{\frac{4\pi}{5}})^{[2]}_{A} \equiv
			\left\{
			\begin{aligned}
				\ket{\psi_{12}} &= \ket{0\theta},   & \ket{\psi_{13}} &= \ket{0+}, \\
				\ket{\psi_{14}} &= \ket{0i_+},      & \ket{\psi_{21}} &= \ket{\theta0}, \\
				\ket{\psi_{23}} &= \ket{\theta+},   & \ket{\psi_{24}} &= \ket{\theta i_+}, \\
				\ket{\psi_{31}} &= \ket{+0},        & \ket{\psi_{32}} &= \ket{+\theta}, \\
				\ket{\psi_{34}} &= \ket{+i_+},      & \ket{\psi_{41}} &= \ket{i_+0}, \\
				\ket{\psi_{42}} &= \ket{i_+\theta}, & \ket{\psi_{43}} &= \ket{i_++}
			\end{aligned}
			\right\}
		\end{equation}
		where $\ket{\theta} = \cos\left(\frac{4\pi}{5}\right)\ket{0} + \sin\left(\frac{4\pi}{5}\right)\ket{1}$.
		\noindent Now we define the sets 
		
		\begin{align}
			S'_1 &= \{\ket{\psi_{12}},\ \ket{\psi_{34}},\ \ket{\psi_{43}}\}, & S'_2 &= \{\ket{\psi_{13}},\ \ket{\psi_{34}},\ \ket{\psi_{41}}\},\nonumber\\
			S'_3 &= \{\ket{\psi_{14}},\ \ket{\psi_{31}},\ \ket{\psi_{43}}\}, & S'_4 &= \{\ket{\psi_{21}},\ \ket{\psi_{34}},\ \ket{\psi_{43}}\}, \\
			S'_5 &= \{\ket{\psi_{23}},\ \ket{\psi_{34}},\ \ket{\psi_{42}}\}, & S'_6 &= \{\ket{\psi_{24}},\ \ket{\psi_{32}},\ \ket{\psi_{43}}\}.\nonumber
		\end{align}
		
		\noindent Each  of the sets $S_i'$ is antidistinguishable since they satisfy \cref{eq:caves1,eq:caves2}. Therefore by Lemma 1 $(\mathcal{S}_{NL}^{\frac{4\pi}{5}})^{[2]}_X$ is antidistinguishable. This implies that simply by measuring on her part, Alice is able to make the set $(\mathcal{S}_{NL}^{\frac{4\pi}{5}})^{[2]}$ locally antidistinguishable . Thus $\mathcal{S}_{NL}^{\frac{4\pi}{5}}$ allows $(2,1)$-LSAM.
		
		To prove that $\mathcal{S}_{NL}^{\frac{4\pi}{5}}$ is not $2-$CLSM, we show that there does not exist $\{\ket{\alpha_i}_A\ket{\beta_i}_B\}_{i=1}^{12}$ for $\ket{\Psi_i}\in (\mathcal{S}_{NL}^{\frac{4\pi}{5}})^{[2]}$ such that $\braket{\alpha_i\beta_i|\Psi_j}=0\,\,\forall i\neq j$ which according to Definition \ref{defCheflesCLSD}, is a necessary and sufficient criteria. 
		
		\noindent Observe that:
		\begin{align}
			\ket{0+} ~ + ~ i\ket{0i_+} ~ &\propto ~ \ket{00}\\
			\ket{0+} ~ - ~~ \ket{0i_+} ~ &\propto ~ \ket{01}\\
			\ket{0+} ~ + ~ i\ket{0i_+} ~ - ~ (1+i)&\ket{+0} ~ \propto ~ \ket{10}\\
			\ket{0+} ~ - ~ (1-i)\ket{0i_+} ~ - ~ (1+i)\ket{+0} ~ &+ ~ (1+i)\ket{+i_+}~ \propto ~ \ket{11}  
		\end{align}
		
		\noindent Thus, there does not exist any state $\ket{\alpha_i}\in \mathbb{C}^4$ orthogonal to all states $\{\ket{0+},\ket{0i_+},\ket{+0},\ket{+i_+} \}$ since, they span the whole $\mathbb{C}^4$. Similarly there does not exist $\ket{\beta_i}\in \mathbb{C}^4$ orthogonal to all states $\{\ket{0+},\ket{0i_-},\ket{+0},\ket{+i_-} \}$ . The global $2-$conclusive markability of the states is assured by the linear indepedence of $(\mathcal{S}_{NL}^{\frac{4\pi}{5}})^{[2]}$ which trivial from the method used in \cite{CLSM,roman2008advanced}.This concludes the proof.
		
	\end{proof}

	We also know that from Lemma \ref{lemmaCheflesGlobal}, although the ensemble $\mathcal{S}^{\uparrow\hspace{-.05cm}\downarrow}$ does not admit CLSD, and also not $2$-CLSM \cite{CLSM}, it is locally antidistinguishable (Proposition \ref{antisic}). Thus it  also allows for LSAM.

	\begin{figure}[t!]
		\centering
		\includegraphics[width=0.5\textwidth]{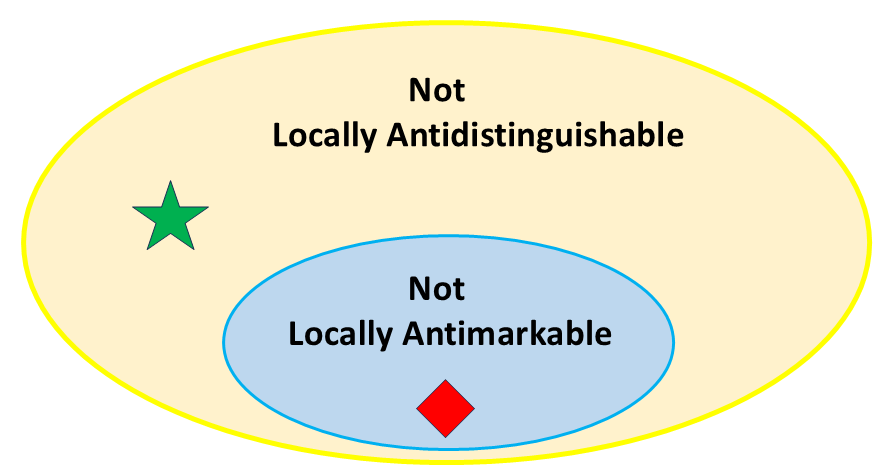}
		\caption{The strict hierarchy between LSAD and LSAM demonstrated in Propositions \ref{Appendix_D_Prop_5},\ref{Appendix_E_Prop_8} is represented by symbol $\bigstar$ which denotes a set of quantum states that is not locally antidistinguishable but is locally $(2,1)$-antimarkable. In contrast, Theorem \ref{P_NL_tri} corresponds to symbol $\blacklozenge$ representing a set of states that are neither LSAD nor (2,1)-LSAM.}\label{LSAMfig2}
	\end{figure}

	\begin{figure}[t!]
		\centering
		\includegraphics[width=0.5\textwidth]{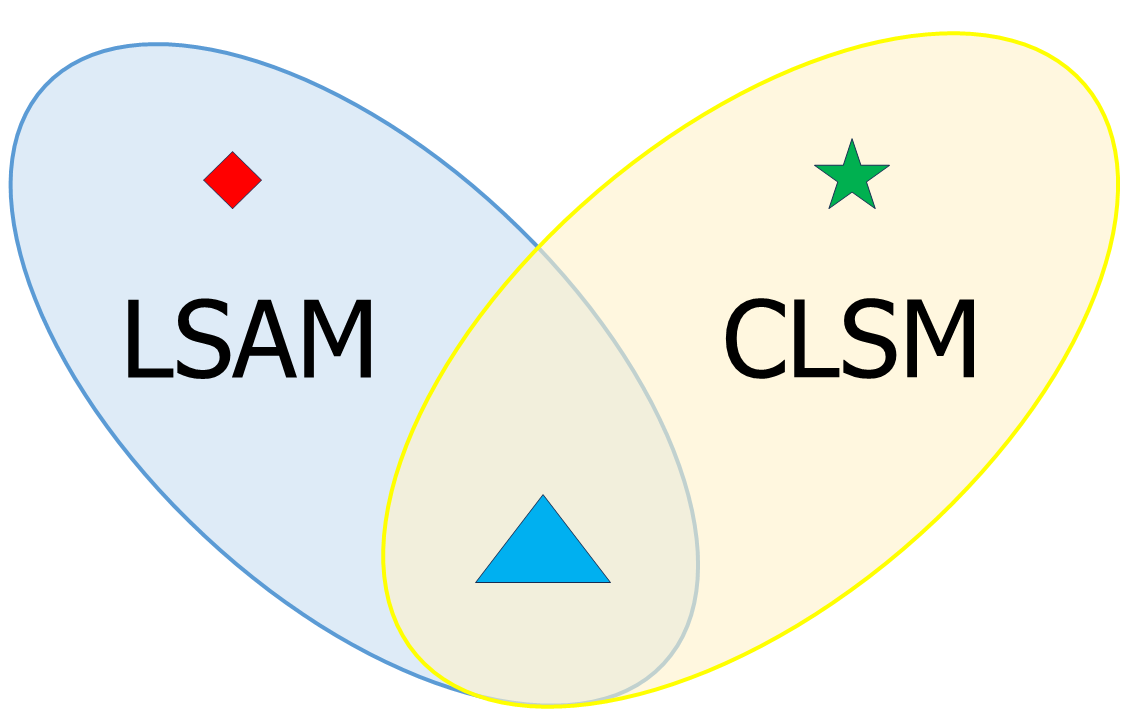}
		\caption{The absence of a hierarchy between LSAM and CLSM was shown in Proposition \ref{Appendix_E_Prop_8} and Theorem \ref{P_NL_tri} represented by symbols $\blacklozenge$ and $\bigstar$ respectively. In contrast as shown in Proposition \ref{prop7}, the set of quantum states $\mathcal{S}_{\text{NL}_1}$ is an example which allows both CLSM and LSAM and is represented as $\blacktriangle$.}\label{LSAMfig3}
	\end{figure}

	\noindent Consider the following set of tripartite product states: 
	\begin{align}\label{NLSAM_tripartite}
		{S}_{NL_2}\equiv\left\{
		\begin{aligned}
			\ket{\psi_1}=\ket{0}&\ket{0}\ket{0}\ ,\ \ket{\psi_2}=\ket{0}\ket{\theta}\ket{\theta},\\
			&\ket{\psi_3}=\ket{\theta}\ket{0}\ket{\theta}
		\end{aligned}
		\right\}
	\end{align}
	where $\ket{\theta} = \cos\theta\ket{0}+\sin\theta\ket{1}$
	
	\begin{theorem}\label{P_NL_tri}
		There exist product states which are nonlocal in the LSAM paradigm.
	\end{theorem}
	\begin{proof}

		The set $\mathcal {S}_{NL_2}$(\ref{NLSAM_tripartite}) of tripartite quantum states, is globally antidistinguishable but it does not allow $(2,1)$-LSAM for $\theta\in\big[\frac{\pi}{4},\cos^{-1}(\frac{1}{\sqrt{3}})\big]\cup\big[\cos^{-1}(-\frac{1}{\sqrt{3}}),\frac{3\pi}{4}\big]$. It also allows CLSD and thus CLSM \cite{CLSM}. Hence, the antimarking paradigm reveals its nonlocality, which was suppressed in the CLSD and the CLSM paradigms.
		
		\noindent Calculating the pairwise overlap of the states in $\mathcal {S}_{NL_2}$ i.e. $x_1=|\braket{\psi_1|\psi_2}|^2,x_2=|\braket{\psi_2|\psi_3}|^2,x_3=|\braket{\psi_3|\psi_1}|^2$ we get $x_1= x_2 = x_3= \cos^4\theta$. So , from \cref{eq:caves1,eq:caves2} we get $\mathcal{S}_{NL_2}$ is antidistinguishable if and only if :
		\begin{align}
			3\cos^4\theta&<1\\
			(1-3\cos^4\theta)^2&\geq4\cos^{12}\theta
		\end{align}
		The above two equations are  simultaneously satisfied in the region $\frac{\pi}{4}<\theta<\frac{3\pi}{4}$
		
		\noindent $\mathcal {S}_{NL_2}$ is $(2,1)$-LSAM if at least one local part $\mathcal {S}_{NL_2}^{[2]}$ is weakly antidistinguishable. The local part for Alice and Bob for $\mathcal {S}_{NL_2}^{[2]}$ is $\{\ket{00},\ket{0\theta},\ket{\theta0}\}$ and that for Charlie is $\{\ket{0\theta},\ket{\theta0},\ket{\theta\theta}\}$ again using the conditions \cref{eq:caves1,eq:caves2} we get they are antidistinguishable if and only if:
		\begin{align}
			2\cos^2\theta+\cos^4\theta&<1\\
			(1-2\cos^2\theta+\cos^4\theta)^2&\geq4\cos^{8}\theta
		\end{align}
		Solving the above two equation we get $\cos^{-1}(\frac{1}{\sqrt3})\leq \theta\leq\cos^{-1}(-\frac{1}{\sqrt3})$ . This conclude the proof that in the region $\theta\in\big[\frac{\pi}{4},\cos^{-1}(\frac{1}{\sqrt{3}})\big]\cup\big[\cos^{-1}(-\frac{1}{\sqrt{3}}),\frac{3\pi}{4}\big]$  $\mathcal {S}_{NL_2}$ is not $(2,1)$-LSAM.

		\noindent It is also easy to see that $\mathcal{S}_{NL_2}$ is locally conclusively distinguishable. The three parties use the following measurement strategy:
		$\mathcal{M}_A=\mathcal{M}_B:=\{\ket{0}\bra{0},\ket{1}\bra{1}\}$ and $\mathcal{M}_C = \{  \ket{\theta}\bra{\theta},\ket{\theta^{\perp}}\bra{\theta^{\perp}}\}$
		
	\end{proof}
	Hence, $\mathcal {S}_{NL_2}$ is locally conclusively distinguishable, which implies that it is locally conclusively markable \cite{CLSM}, but it shows a strong  form of nonlocality when it comes to antimarkability. Hence, this also establishes the inequivalence of CLSM and LSAM.

\end{document}